\documentclass[12pt]{amsart}
\usepackage{fullpage}
\usepackage{tikz,wasysym}
\usepackage{pgf}
\usetikzlibrary{positioning,automata}
\usepackage[font=small,labelfont=bf]{caption}
\usepackage{bbm}
\usepackage{verbatim}
\usepackage[linktocpage]{hyperref}
\hypersetup{
    colorlinks=true,        
    linkcolor=red,          
    citecolor=blue,         
    filecolor=magenta,      
    urlcolor=cyan           
}

\usepackage{amssymb}
\usepackage{url}
\usepackage{adjustbox}
\usepackage{ae,amsfonts,euscript,enumerate}
\usepackage{amsmath,euscript}
\usepackage{amsthm}
\usepackage{comment}
\usepackage{amssymb,array}
\usepackage{enumitem}
\usepackage{booktabs}

\usepackage{url}
\usepackage{graphics}
\usepackage{latexsym}
\usepackage{epsf}

\newtheorem{theorem}{Theorem}[section]
\newtheorem{lemma}[theorem]{Lemma}
\newtheorem{corollary}[theorem]{Corollary}

\newtheorem{proposition}[theorem]{Proposition}

\theoremstyle{definition}

\newtheorem{notation}[theorem]{Notation}

\begin{document}

\author{Jason P. Bell}
\address{Department of Pure Mathematics \\
University of Waterloo\\
Waterloo, ON  N2L 3G1 \\
Canada}


\thanks{The author was supported by NSERC Grant RGPIN-2016-03632.}

\title{The upper density of an automatic set is rational}

\begin{abstract} Given a natural number $k\ge 2$ and a $k$-automatic set $S$ of natural numbers, we show that the lower density and upper density of $S$ are recursively computable rational numbers and we provide an algorithm for computing these quantities.  In addition, we show that for every natural number $k\ge 2$ and every pair of rational numbers $(\alpha,\beta)$ with $0<\alpha<\beta<1$ or with $(\alpha,\beta)\in \{(0,0),(1,1)\}$ there is a $k$-automatic subset of the natural numbers whose lower density and upper density are $\alpha$ and $\beta$ respectively, and we show that these are precisely the values that can occur as the lower and upper densities of an automatic set.  \end{abstract}

\maketitle
\section{Introduction}
Given a subset $S$ of the natural numbers, a natural question to ask is: \emph{What proportion of natural numbers lie in $S$?}
To answer this, one lets $\pi_S(x)$ denote the number of elements of $S$ that are less than $x$ and one then studies how $\pi_S(x)/x$ behaves as $x$ tends to infinity.  In general, $\lim_{x\to\infty} \pi_S(x)/x$ need not exist, but when it does, we call the limit the \emph{density} of the set of $S$.  To deal with the fact that densities of sets of natural numbers need not exist, one can instead consider the \emph{upper} and \emph{lower densities} of the set $S$, given respectively by $\limsup_{x\to\infty} \pi_S(x)/x$ and $\liminf_{x\to\infty} \pi_S(x)/x$, which together provide a rough answer to the motivating question asked above. 

The lower density and upper density of a set of natural numbers are real numbers $\alpha$ and $\beta$ in $[0,1]$ with $\alpha\le \beta$, and given $(\alpha,\beta)$ satisfying these conditions there exists a set whose lower and upper densities are $\alpha$ and $\beta$, respectively.  When one restricts one's attention to so-called automatic sets---that is, subsets of the natural numbers whose elements are precisely those whose base-$k$ expansions are accepted by a finite-state machine for some $k\ge 2$---the questions of density become significantly more constrained.  For example, when $S$ is an automatic set whose density exists, a result of Cobham \cite{Cobham} shows that the density is necessarily rational.  We give further background on automatic sets and automata in \S\ref{background}.

In general, the density of an automatic set of natural numbers need not exist.  As an example,
let $S$ denote the set of numbers whose base-$k$ expansion has even length. If we let $\pi_S(x)$ denote the number of elements of $S$ that are less than $x$, then 
$$\pi_S(k^{2n}) = 1+\sum_{j=1}^n (k^{2j}-k^{2j-1}) \sim \left(\frac{k}{k+1}\right)\cdot k^{2n},$$ while
$$\pi_S(k^{2n+1}) = 1+\sum_{j=1}^n (k^{2j}-k^{2j-1}) \sim \left(\frac{1}{k+1}\right)\cdot k^{2n+1},$$ and so the lower density of $S$ is at most $1/(k+1)$ and the upper density is at least $k/(k+1)$.  (In \S\ref{computable} we give a more general construction that shows that $1/(k+1)$ and $k/(k+1)$ are respectively the lower and upper densities of $S$.)  One can nevertheless ask what one can say about the upper and lower densities of an automatic set.  In this paper we answer this question completely, showing the lower and upper densities are recursively computable rational numbers in $[0,1]$ and, moreover, we characterize exactly which pairs $(\alpha,\beta)$ can be realized as the lower density and upper density of an automatic set.  

The first part of this characterization is given by the following more general result concerning automatic sequences.  

\begin{theorem} \label{thm:main} Let $k\ge 2$ be a natural number and let $h:\mathbb{N}\to \Delta\subseteq \mathbb{Q}_{\ge 0}$ be a $k$-automatic sequence and let $s(n)=\sum_{j<n} h(n)$.  Then $$\limsup_{n\to \infty} s(n)/n$$ is a recursively computable rational number.
\end{theorem}
We make some remarks about what is meant by `recursively computable' in the statement of Theorem \ref{thm:main}.  Given a $k$-automatic sequence $h$, one can build a deterministic finite automaton with output $\Gamma$, which takes the base $k$-expansion of $n$ as input, reading left to right, and gives $h(n)$ as output.  We give an algorithm that allows one to determine the limsup of $s(n)/n$ from the machine $\Gamma$ (see \S\ref{computable} for further details).  

As an immediate corollary of Theorem \ref{thm:main}, we obtain the following result.
\begin{corollary} \label{cor:main2} Let $k\ge 2$ be a natural number and let $S\subseteq \mathbb{N}$ be a $k$-automatic set.  Then the upper and lower densities of $S$ are recursively computable rational numbers.  
\end{corollary}
We point out that Corollary \ref{cor:main2} can be seen as an extension of the result of Cobham \cite{Cobham} mentioned earlier about the density of an automatic set of natural numbers, when it exists.  Since the density exists if and only if the upper and lower densities coincide, this follows from Corollary \ref{cor:main2}.  As the example given in which $S$ is the set of numbers whose base $k$-expansion has even length illustrates, our result is strictly stronger than Cobham's result.

In light of Corollary \ref{cor:main2}, it is natural to ask which possible pairs $(\alpha,\beta)\in \mathbb{Q}^2$ can occur as the lower and upper densities of an automatic set.  We are able to completely characterize which pairs can occur. 
\begin{theorem}\label{thm:main3}
 Let $k\ge 2$ be a positive integer and let $(\alpha,\beta)$ be a pair of rational numbers satisfying either $0<\alpha\le \beta <1$ or $(\alpha,\beta)\in \{(0,0),(1,1)\}$.  Then there is a $k$-automatic set $S$ whose lower density and upper density are $\alpha$ and $\beta$ respectively.  Conversely, if $S$ is a $k$-automatic set whose lower density is $\alpha$ and whose upper density is $\beta$ then either $(\alpha,\beta)\in \{(0,0),(1,1)\}$ or $\alpha,\beta$ are rational numbers with $0<\alpha\le \beta<1$.
 \end{theorem}

The outline of this paper is as follows.  In \S\ref{background}, we give some of the basic background on automatic sequences and sets and related notation we will make use of. In \S\ref{proof}, we give an overview of the strategy used to proof Theorem \ref{thm:main} and then prove Theorem \ref{thm:main}. In \S\ref{computable} we provide an algorithm that allows one to compute the upper and lower density of an automatic set.  In \S\ref{example} we give examples that are used to demonstrate Theorem \ref{thm:main3}, and finally in \S\ref{conc} we give some concluding remarks and raise a question concerning possible extensions of Theorem \ref{thm:main} to morphic sequences.
\section{Background on automata and automatic sets}\label{background}
In this section we give the necessary background on finite-state automata and $k$-automatic sequences and sets. 

Let $\Sigma$ be a nonempty finite set and let $\Sigma^*$ denote the free monoid on $\Sigma$. A \emph{deterministic finite automaton with output} (DFAO) is a $6$-tuple $$\Gamma = (Q, \Sigma, \delta, q_0, \Delta, \tau),$$ where $Q$ is a finite set of states, $\Sigma$ is a finite input alphabet, $\delta$ is the transition function from $\Sigma\times Q$ to $Q$, $q_0 \in Q$ is the initial state, $\Delta$ is an output alphabet, and $\tau$ is the output function from $Q$ to $\Delta$.  Less formally, a DFAO is simply a directed graph in which the vertices are the elements of $Q$; for each vertex $q\in Q$ and each $x\in \Sigma$ we have a directed arrow with label $x$ from $q$ to the state $\delta(x,q)$.  Given a word $w\in \Sigma^*$, the DFAO gives us an output in $\Delta$ as follows: we begin at the initial state $q_0$ and then, reading $w$ from left to right, we obtain a path in this directed graph by moving vertex to vertex as we read the letters of $w$.  After reading $w$ we end up at some state $q\in Q$ and we then apply $\tau$ to obtain an output in $\Delta$.  Thus we can associate a map $f:\Sigma^*\to \Delta$ with a DFAO $\Gamma$.  

We give an example of a DFAO in Figure 1 that generates the map $f$ from $\{0,1\}^*$ to $\{0,1\}$ and $f$ is $1$ precisely when the string $w$ is either of the form $0^a$ or of the form $0^a 1 u$ where $u$ has odd length.  In particular, $f$ induces a well-defined map $h:\mathbb{N}\to \{0,1\}$ given by taking the binary expansion of $n$ and applying $f$; then $h(n)$ is $1$ precisely when the number of digits in the binary expansion of $n$ is even and is $0$ otherwise, where we take the binary expansion of $0$ to be the empty word.

\begin{figure}[!htbp]
\begin{tikzpicture}[shorten >=1pt,node distance=2cm,on grid]
  \node[state,initial]   (q_0)                {$q_0/1$};
  \node[state]           (q_1) [right=of q_0] {$q_1/0$};
   \node[state]           (q_2) [right=of q_1] {$q_2/1$};
  \path[->] (q_0)		(q_0) edge [loop above]   node         {0} ()
   (q_0)		 	edge  [bend left]   	node [above] {1} (q_1)
		(q_1)		edge [bend left] node [above] {0,1} (q_2)
		(q_2)		edge [bend left] node [below] {0,1} (q_1);
  \end{tikzpicture}
  \caption{The DFAO generating the sequence $h(n)$.}
\end{figure}
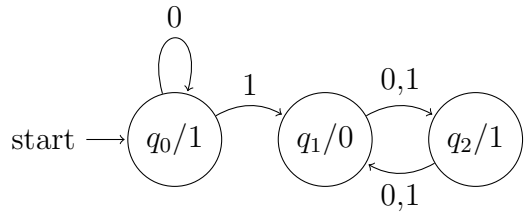

 Let $k\ge 2$ be a natural number and let $\Sigma_k$ be the alphabet $\{0,1,\ldots ,k-1\}$.  For every natural number $n$, there is a word $w=(n)_k\in \Sigma_k^*$, which is the base-$k$ expansion of $n$, where we define $(0)_k$ to be the empty word; conversely, given a non-empty word $w\in \Sigma_k^*$ with no leading zeros there is a natural number $n=[w]_k$, which is the natural number whose base-$k$ expansion is $w$. In the case when $w$ is the empty word, we take $[w]_k=0$. A sequence $a: \mathbb{N} \rightarrow \Delta$ is called \emph{$k$-automatic} if there exists a DFAO $\Gamma=(Q, \Sigma_k, \delta, q_0, \Delta, \tau)$ such that for each $n \in \mathbb{N}$, $a(n)$ can be computed from $\Gamma$ by feeding the word $(n)_k$ into $\Gamma$, reading the digits from left to right. We then say that a subset $S \subseteq \mathbb{N}$ is a \emph{$k$-automatic set} if the characteristic function of $S$, $\chi_S:\mathbb{N}\to \{0,1\}$ defines a $k$-automatic sequence.  It is worth noting that some authors define $k$-automatic sequences using the convention that one reads the base-$k$ expansion of a number $n$ from right to left.  As it turns out, this definition is equivalent to the one we give (cf. \cite[Corollary 4.3.5]{AS}).  
 
 An important property of an automatic sequence $f:\mathbb{N}\to \Delta$, when regarded as a map from $\Sigma_k^*$ to $\Delta$ with the property that $f(0w)=f(w)$ for every word $w\in \Sigma_k^*$, is that there is a finite set of maps $f=f_1,\ldots ,f_d:\Sigma_k^*\to\Delta$ with the property that for each $i\in \Sigma_k^*$ and each $j\in \{1,2,\ldots ,d\}$ there is some $\ell\in \{1,2,\ldots ,d\}$ such that $f_j(iw)=f_{\ell}(w)$ for every word $w\in \Sigma_k^*$, where $iw$ is the concatenation of $i$ and $w$.  We call these maps the (left) $k$-\emph{kernel} of $f$.

We make use of the following fact about automatic sequences, which we suspect is well-known, although we are unaware of a proof in the literature.
\begin{proposition} Let $k\ge 2$ be a natural number, let $h:\mathbb{N}\to \mathbb{Q}_{\ge 0}$ be a $k$-automatic sequence, and let $s(n)=\sum_{j<n} h(j)$.  Then there exist $\beta\in (0,k)$, $C>0$, $a\ge 1$, and rational numbers $c_{j}$ for $j\in  \{0,1,\ldots ,a-1\}$ such that 
$$|s(k^{an+j}) - c_{j} k^{an+j}|<C \beta^{an}$$ for every $n\ge 0$.  Moreover, $a$ and the rational numbers $c_{0},\ldots , c_{a-1}$ are recursively computable and $\beta$ can be effectively determined. 
\label{thm:perron}
\end{proposition}
\begin{proof}
We have an automatic map $f$ from $\Sigma_k^*$ to $\mathbb{Q}$ with $f(0w)=f(w)$ and $f((n)_k)=h(n)$, and we let $f=f_1,\ldots ,f_s$ denote the left $k$-kernel of $f$.  Then there are $s\times s$ matrices $A_0,\ldots ,A_{k-1}$ with entries in $\{0,1\}$ such that each $A_j$ has exactly one $1$ in each row and such that for each $w\in \Sigma_k^*$ and each $i\in \{0,\ldots ,k-1\}$ we have  
$$[f_1(iw),\ldots , f_s(iw)]^T = A_i [f_1(w),\ldots ,f_s(w)]^T.$$
Let $B=A_0+\cdots +A_{k-1}$ and let $v_0=[f_1(\epsilon),\ldots ,f_s(\epsilon)]^T$.
Then \begin{align*}
e_1^T B^n v_0 &= \sum_{0\le i_0,\ldots ,i_{n-1} <k} e_1^T A_{i_{n-1}}\cdots A_{i_0} v_0\\
&= \sum_{\{w \colon |w|=n\}} f(w) \\
&= s(k^n).
\end{align*}
Now $B$ is a sum of $k$ matrices, each of which have exactly one $1$ in each row, and the entries of $B$ are nonnegative.  In particular, the row sums of $k^{-1}\cdot B$ are all $1$ and the entries are nonnegative. By a result of Karpelevi\u{c} \cite{Karp} (see also Higham and Lin \cite{HL}) we have:
\begin{enumerate}
\item[(i)] $k$ is an eigenvalue of $B$;
\item[(ii)] each eigenvalue of $B$ of modulus $k$ is of the form $k\omega$ with $\omega$ a root of unity;
\item[(iii)] each eigenvalue of $B$ has modulus at most $k$.
\end{enumerate}
Then $$s(k^n)=e_1^T B^n v_0 = \sum_{i=1}^m p_i(n) \gamma_i^n,$$ where $\gamma_1,\ldots ,\gamma_m\in \bar{\mathbb{Q}}$ are the eigenvalues of $B$ and $p_1,\ldots ,p_m\in \bar{\mathbb{Q}}[x]$ are polynomials of degree at most one less than the largest Jordan block with eigenvalue $\gamma_i$ occurring in the Jordan form of $B$.  We may assume that $\gamma_1,\ldots ,\gamma_t$ are the eigenvalues of modulus $k$ (and hence $\gamma_i=k\omega_i$ with $\omega_i$ a root of unity for $i=1,\ldots ,t$) and we may pick $\beta\in (0,k)$ such that $\gamma_{t+1},\ldots ,\gamma_m$ have modulus strictly less than $\beta$. Now we pick $a>0$ so that $\omega_1^a=\cdots = \omega_t^a=1$.
For $j=0,\ldots , a-1$ we have
$$s(k^{an+j}) =  q_j(n) k^{an+j} + \sum_{i=t+1}^m p_{i,j}(n) \gamma_i^{an+j}$$ for some polynomials 
$q_j, p_{i,j}$ in $\bar{\mathbb{Q}}[x]$. A priori we only know that $q_j(x)$ is a polynomial with coefficients in $\bar{\mathbb{Q}}$, but we shall show that it is in fact a constant polynomial, where the constant is rational. 
 
Let $K$ be the Galois closure of the number field generated by the coefficients of $p_{1},\ldots ,p_m$ and $\gamma_1,\ldots ,\gamma_m$, and let $G$ be the Galois group of $K$. If $\sigma$ is in $G$, then $\sigma$ permutes $\gamma_1,\ldots ,\gamma_m$, but since $\sigma$ fixes $k$ and takes roots of unity to other roots of unity, we see that $\sigma$ permutes $\gamma_1,\ldots ,\gamma_t$; consequently, $\sigma$ permutes $\gamma_{t+1},\ldots ,\gamma_m$.  From these facts, we first get 
$$\sum_{\sigma\in G} \sigma(s(k^{an+j})) = |G| s(k^{an+j}),$$ as $s(k^{an+j})$ is rational.  On the other hand, this is equal to 
$$\left( \sum_{\sigma\in G} \sigma(q_j(n)) \right) k^{an+j} + \sum_{\sigma\in G} \sum_{i=t+1}^m \sigma(p_i(n)) \sigma(\gamma_i)^n.$$ 
Let $h_j(n):=\sum_{\sigma\in G} \sigma(q_j(n))$.  Then by construction $h_j(n)$ is a polynomial in $n$ and by construction it is fixed by every element of $G$ and hence $h_j(n)$ is a rational number for every $n$; moreover, since $\sigma$ permutes $\gamma_{t+1},\ldots ,\gamma_m$, they each have modulus strictly less than $\beta$.  

Thus $$s(k^{an+j}) = |G|^{-1} h_j(n) k^{an+j} + {\rm O}(\beta^{an+j}).$$  Finally observe that $|s(k^{an+j})|\le C_0 k^{an+j}$ where $C_0$ is the maximum of the absolute values of the range of $f$.  It follows that $h_j(n)$ is a rational constant for $j=0,\ldots ,a-1$, and we let 
$c_j\in \mathbb{Q}$ denote the value $|G|^{-1} h_j(0)$.  We have thus proved the rationality of the constants $c_j$.

We now make remarks concerning the computability of $a$ and $c_0,\ldots ,c_{a-1}$.  Observe that the $d\times d$ matrix $B$ can be computed from the DFAO giving $f$, and if $\omega$ is a $b$-th root of unity such that $k\omega$ is a root of the characteristic polynomial of $B$, then the field extension $[\mathbb{Q}(\mathbb{\omega}):\mathbb{Q}]$ has degree at most $d$.  It follows that $\phi(b)\le d$, where $\phi(n)$ is Euler's $\phi$-function.  If $\phi(b)\le d$ then each prime power factor $p^r$ of $b$ must satisfy $p^{r-1}(p-1)\le d$, and so 
$b$ divides $\prod_{p\le d} p^{\lfloor \log_p(d)\rfloor}$.  Then we can take $a$ to be the quantity $\prod_{p\le d} p^{\lfloor \log_p(d)\rfloor}$.  

We next explain how to compute the rational numbers $c_0,\ldots ,c_{a-1}$.
By our choice of $a$, the matrix $B^a$ has an eigenvalue $k^a$ and the remaining eigenvalues are strictly less than $k^a$ in modulus.  Moreover, we have shown that the Jordan blocks in $B^a$ associated to the eigenvalue $k^a$ are all of size one and so the minimal polynomial of $B^a$ has $k^a$ as a simple root.  We can compute the characteristic polynomial of $B^a$ and we can find the factorization of this polynomial in $\mathbb{Q}[x]$ into irreducible factors (see \cite{LLL} for details) and hence we can find all of the finitely many monic rational polynomial factors of the characteristic polynomial of $B^a$.  Then by applying each of these factors to $B^a$, we can compute the minimal polynomial of $B^a$.

As mentioned before, the minimal polynomial has $k^a$ as a simple root and so if $q(x)$ denotes the minimal polynomial of $B^a$ then we have $q(x)=(x-k^a)q_0(x)$, where all roots of $q_0(x)$ are strictly less than $k^a$ in modulus.  Since $k^a$ is not a root of $q_0(x)$, we can then use the division algorithm for polynomials to compute a rational polynomial $t(x)$ and a nonzero rational constant $\lambda$ such that $q_0(x) = t(x)(x-k^a) + \lambda$.  
For $j\in \{0,1,\ldots ,a-1\}$ we have
$$s(k^{an+j})=e_1^T B^{an+j} v_0.$$  To compute $c_j$, we let $v_{j}':=\lambda^{-1} q_0(B) B^j v_0$ and we let 
$v_{j}'':=-\lambda^{-1} t(B) (B-k^a I)B^j v_0$, both of which are computable.  
By construction, $(B-k^a I)v_{j}' =0$ since $(x-k^a)q_0(x)$ is the minimal polynomial of $B$.  Similarly, $q_0(B)v_{j}''=0$; moreover, $v_{j}'+v_{j}''=B^j v_0$ by construction.  Then $v_{j}''$ lies in the sum of the eigenspaces corresponding to the roots of $q_0(x)$ and since these are all strictly less than $k^a$ in modulus, we see that
$$e_1^T B^{an} v_{j}'' = {\rm o}(k^{an})$$ as $n\to\infty$.  On the other hand, $B^a v_{j}' = k^a v_{j}$ and so 
$$e_1^T B^{an} v_{j}'= k^{an} e_1^T v_{j}'.$$  Then we can compute the rational number $c_j:=e_1^T v_{j}'$ and by construction
$$e_1^T B^{an+j}v_0 = e_1^T B^{an} (v_j'+v_j'') = c_j k^{an} + {\rm o}(k^{an})$$ and so we have shown that the $c_{j}$ are recursively computable. Finally, since we can numerically calculate the eigenvalues of $B$ to arbitrary precision, we can determine some $\beta\in (0,k)$ that is strictly larger than the eigenvalues that are strictly less than $k$ in modulus. 
\end{proof}

\section{Proof of rationality}\label{proof}
In this section we prove Theorem \ref{thm:main}.  To do so, we need a technical lemma, which we now prove.
 \begin{lemma}
\label{limsup} Let $k\ge 2$ be a natural number, let $\gamma$ be a positive real number, let $s_n, s_n'$ be sequences of positive numbers, let $u', v', u,v$ be positive real numbers, and let $b$ and $c$ be positive integers.  If
$$\frac{(v' k^{b+c} + u' k^c + s_n')}{(v k^{b+c} + u k^c + s_n)}\to \gamma$$ as $n\to\infty$ and 
$$\limsup_{n\to\infty} \frac{(v' k^c + s_n')}{(v k^c + s_n)} \le \gamma,$$ then
$$\limsup_{n\to\infty} \frac{(v' k^{2b+c} + u'(k^{b+c}+k^c)+s_n')}{(vk^{b+2c} + u (k^{b+c}+k^c) + s_n)} \ge \gamma.$$ 
\end{lemma}
\begin{proof}
Let $$X_n=v' k^c + s_n',~~X'_n=v' k^{b+c} + u' k^c + s_n', ~{\rm and}~~X_n''=  v' k^{2b+c} + u'(k^{b+c}+k^c)+s_n'.$$  Similarly, we let 
$$Y_n=v k^c + s_n,~~Y_n'=v k^{b+c} + u k^c + s_n,~{\rm and}~~Y_n''=  v k^{2b+c} + u(k^{b+c}+k^c)+s_n.$$
Then $X_n''=(k^b+1)X_n'-k^b X_n$ and $Y_n'' = (k^b+1) Y_n'-k^b Y_n$ and our goal is to show that
$\limsup_{n\to\infty} X_n''/Y_n''\ge \gamma$.

Now suppose that the conclusion to the statement of the lemma does not hold.  Then there exists some $\epsilon>0$ with $\gamma-\epsilon>0$ such that
$$X_n''/Y_n'' < \gamma -\epsilon$$ for all $n$ sufficiently large. 
In other words, for $n$ large we have 
$$(k^b+1)X_n'-k^b X_n < (\gamma-\epsilon)((k^b+1) Y_n'-k^b Y_n).$$
We may rewrite this inequality as
\begin{equation}
\label{eq:XY''}
(k^b+1)X_n' +  (\gamma-\epsilon) k^b Y_n < (\gamma-\epsilon)(k^b+1)Y_n' +k^b X_n,
\end{equation} for $n$ sufficiently large.
By assumption, $X_n'/Y_n'\to \gamma$ for $n$ large. Thus 
\begin{equation}
X_n' \ge (\gamma- \epsilon/(3k^b))Y_n'
\label{eq:XY}
\end{equation}
 for all $n$ sufficiently large.  Similarly, since $\limsup X_n/Y_n \le \gamma$, we have
 \begin{equation}
 \label{eq:XY'}
X_n \le (\gamma+ \epsilon/(3k^b))Y_n
\end{equation}
for $n$ sufficiently large.
Then using Equations (\ref{eq:XY}) and (\ref{eq:XY'}) along with Equation (\ref{eq:XY''}), we have
 \begin{align*}
 (k^b+1)(\gamma- \epsilon/(3k^b))Y_n' + k^b (\gamma-\epsilon) Y_n 
&\le (k^b+1)X_n' + k^b (\gamma-\epsilon) Y_n \\ &< (\gamma-\epsilon)(k^b+1)Y_n' +k^b X_n \\
 &< (\gamma-\epsilon)(k^b+1)Y_n' + k^b (\gamma+ \epsilon/(3k^b))Y_n.
 \end{align*}
 for $n$ sufficiently large.
Then by rearranging the inequality
 $$ (k^b+1)(\gamma-\epsilon/(3k^b))Y_n' + k^b (\gamma-\epsilon) Y_n <  (\gamma-\epsilon)(k^b+1)Y_n' + k^b (\gamma+\epsilon/(3k^b))Y_n,$$ we see
 $$(k^b+1)\epsilon(1-1/(3k^b)) Y_n' <  k^b Y_n \epsilon (1+1/(3k^b))$$
 for all $n$ sufficiently large.  Since $\epsilon>0$, we then must have
 \begin{equation}
 \label{eq:YYYYYY}
 (k^b+1)(1-1/(3k^b))Y_n' < k^b  (1+1/(3k^b))Y_n\end{equation} for all sufficiently large $n$.
 
 To complete the proof, observe that since $b$ is a positive integer,
$$
 (k^b+1)(1-1/(3k^b)) = k^b + 2/3 - 1/(3k^b)
 \ge k^b + 1/3 = k^b(1+1/(3k^b)),$$
and so using Equation (\ref{eq:YYYYYY}), we see
 $$k^b(1+1/(3k^b))Y_n'\le  (k^b+1)Y_n' (1-1/(3k^b)) < k^b (1+1/(3k^b))Y_n,$$ or equivalently that $Y_n'\le Y_n$ for all sufficiently large $n$. But it is immediate that $Y_n' \ge Y_n + uk^c > Y_n$ for all $n$, and so we get a contradiction.  
 The result follows.
 \end{proof}

This result is key to obtaining the proof of Theorem \ref{thm:main}.  We will find it useful to make use of certain assumptions and to fix notation for the remainder of the proof.
\begin{notation} Throughout the remainder of the paper we make the following assumptions and notation.\label{notn:1}
\begin{enumerate}
\item We let $\Sigma_k=\{0,1,\ldots ,k-1\}$ and for $w\in \Sigma_k^*$, we let $|w|$ denote the length of $w$ and we let $\prec$ denote the lexicographic order on $\Sigma_k^*$ where $0<1<\cdots <k-1$.
\item We let $h:\mathbb{N}\to \mathbb{Q}$ be $k$-automatic and let $f:\Sigma^*\to \mathbb{Q}$ be a map associated to a DFAO with the property that $f(0w)=f(w)$ and $f(w)=h([w]_k)$.
\item We let $f=f_1,\ldots , f_d:\Sigma_k^*\to \mathbb{Q}$ denote the maps in the left $k$-kernel of $f$.
\item For each $i\in \{1,\ldots ,d\}$ and each $u\in \Sigma_k^*$, we let $\delta(i,u)\in\{1,\ldots ,d\}$ denote the number $j$ with the property that
$f_i(uw) = f_j(w)$ for every $w\in \Sigma_k^*$. 
\item For $w\in \Sigma_k^*$ we let $s_i(w) = \sum_{\{v\prec w, |v|=|w|\}} f_i(v)$ for $i=1,\ldots ,d$ and we let $s(n)=s_1((n)_k)$.
\item We let $\gamma$ denote $\limsup_{n\to\infty} s(n)/n$. 
\item Appealing to Proposition \ref{thm:perron}, we have $a\ge 1$, rational numbers $c_{i,j}$, $1\le i\le d$, $0\le j<a$, and $\beta\in (0,k)$ be such that
$$s_i((k^{an+j})_k)= c_{i,j} k^{an+j} + {\rm O}(\beta^{an}).$$ 
\item For each word $u\in \Sigma_k^*$ with no leading zeros, whose length is a multiple of $a$, and each $j\in \{0,\ldots ,a-1\}$, we let 
$$\gamma_j(u):= \limsup_{\{w\colon a\mid (|w|-j)\}} s([uw]_k)/[uw]_k.$$  
\end{enumerate}
\end{notation}
Using the assumption and notation above, our goal in Lemma \ref{lem:01} is to show that one can decide when $\gamma$ is zero, so for the remainder of this section we assume that $\gamma$ is strictly positive.  We pick words $w_1, w_2, \ldots \in \Sigma_k^*$ with the property 
$$s([w_m]_k)/[w_m]_k \to \gamma.$$ There is some $j\in \{0,1,\ldots ,a-1\}$ such that $|w_m| \equiv j~(\bmod ~a)$ for infinitely many $m$, and by refining our collection if necessary, we may assume that 
$|w_m|\equiv j~(\bmod ~a)$ for every $m$.  In other words $\gamma=\gamma_j(\epsilon)$, where $\epsilon$ is the empty word.  Observe that for each $\ell\ge 0$ there is at least one word of length $\ell a$ that is a prefix of infinitely many $w_m$, and for such words we have $\gamma_j(u)=\gamma$.

\begin{notation} In addition to the items in Notation \ref{notn:1}, we make the following assumptions and fix additional notation.\label{notn:2}
\begin{enumerate}
\item[(9)] We assume $j\in \{0,1,\ldots ,a-1\}$ and $u_1,u_2,\ldots u_{2d+1}$ are words of length $a$ in $\Sigma^*$ such that $u_1$ has no leading zeros and such that $\gamma_j(u_1\cdots u_{2d+1})=\gamma$.
\item[(10)] We let $\mathcal{S}$ denote the set of words $u$ with no leading zeros and whose length is a multiple of $a$ such that $\gamma_j(u)=\gamma$.
\item[(11)] We let $p_1,p_2,p_3\in \{1,\ldots ,2d+1\}$ with $ p_1<p_2<p_3$ be such that 
$\ell:=\delta(1, u_1\cdots u_{p_1})=\delta(1, u_1\cdots u_{p_2})=\delta(1, u_1\cdots u_{p_3})$.  (Such $p_1,p_2,p_3$ exist by the pigeonhole principle.)
\item[(12)] We let $U=u_1\cdots u_{p_1}$, $V=u_{p_1+1}\cdots u_{p_2}$, and $W= u_{p_2+1}\cdots u_{p_3}$.
\end{enumerate}
\end{notation}
Now using the notation above, we explain how the remainder of the proof goes. Using the fact that $\ell=\delta(1,U)=\delta(1,UV)=\delta(1,UVW)$, we see $\ell=\delta(\ell,V)=\delta(\ell,W)$.  By construction $\gamma=\gamma_j(UVW)$ and $\gamma\ge \gamma_j(UW)$.  We use Lemma \ref{limsup} to show that $\gamma_j(UV^2W)\ge \gamma_j(UVW)\ge \gamma_j(UW)$.  But since 
$\gamma_j(UV^2W)\le \gamma$, we get that $\gamma_j(UV^2W)=\gamma$.  We then use an induction argument to show that 
$\gamma_j(UV^m W)=\gamma$ for every $m\ge 1$.  Finally, we show that as $m\to\infty$, $\gamma_j(UV^mW)$ tends to a recursively computable rational number, which then gives the result.  Since there are only finitely many possibilities for words $U$ and $V$ with the length of $UV$ bounded by by $(2d+1)a$, we can compute the values $\gamma_i(AB)$ over all possible triples $(A,B,i)$, with $A,B$ non-trivial words in $\Sigma^*$ such that $A$ has no leading zeros, $|A|+|B|\le (2d+1)a$ and $i\le d$, and the maximum will be the limsup of $s(n)/n$ as $n\to\infty$.

A key step in this strategy is the following lemma, which is inspired by work of Schaeffer and Shallit \cite{SS}.

\begin{lemma} Adopt the assumptions and notation of Notation \ref{notn:1} and Notation \ref{notn:2}. Suppose that $A, B, C\in \Sigma_k^*$ are non-trivial words whose lengths are multiples of $a$ such that $A$ has no leading zeros, $ABC\in \mathcal{S}$, and $e:=\delta(1,A)=\delta(1,AB)=\delta(1,ABC)$. Then $AB^2C$ in $\mathcal{S}$ and $$\gamma(AB^2C)=\gamma(ABC)=\gamma.$$
\label{lem:key}
 \end{lemma}
\begin{proof}
We have $\gamma_j(ABC)=\gamma$ and $\gamma_j(AC)\le \gamma$.  We claim that $\gamma(AB^2C)\ge \gamma(ABC)$, which will then give the result since $\gamma(AB^2C)\le \gamma$. 
By assumption, there is a sequence of words $w_n$ with $|w_n|\equiv j\, (\bmod\, a)$ and $|w_n|\to\infty$ such that $$s([ABCw_n]_k)/[ABCw_n]_k\to \gamma.$$
 Then
\begin{equation}
\label{eq:s}
s([ABCw_n]_k) = \sum_{\{v\colon |v|=|A|,~v\prec A\}}\sum_{\{|w|=|BCw_n|\}} f(vw) + \sum_{\{w\colon |w|=|BCw_n|, w\preceq BCw_n \}} f(Aw),
\end{equation} where $\preceq$ is (pure) lexicographic ordering on $\{0,1,\ldots ,k\}^*$ and we use $\prec$ for strict inequality.  
By Proposition \ref{thm:perron} there is a rational number $\kappa$, $\beta\in (0,k)$, and a positive constant $C_0$ such that 
\begin{equation}
\label{eq:1}
\left| \sum_{\{v\colon |v|=|A|,~v\prec A\}}\sum_{|w|=|BCw_n|} f(vw) - \kappa k^{|BCw_n|} \right| <C_0 \beta^{|w_n|}
\end{equation}
for every $n$, where 
\begin{equation}\label{eq:kappa}
\kappa =  \sum_{\{v\colon |v|=|A|,~v\prec A\}} c_{\delta(1,v),j}\in \mathbb{Q}.
\end{equation}  
Now from the definition of $e$ we have
$$f(Aw)=f(ABw)=f(ABCw) = f_e(w)$$ for every $w$.  Furthermore $f_e(Bw)=f_e(w)$ for every $w$.  
We have
\begin{align*}
&~ \sum_{\{w\colon  |w|=|BCw_n|, w\preceq BCw_n \}} f(Aw) \\ 
&= \sum_{\{w\colon  |w|=|BCw_n|, w\preceq BCw_n \}} f_e(w) \\ 
&= \sum_{\{v\colon |v|=|B|,~v\prec B\}}\sum_{|w|=|Cw_n|} f_e(vw) + \sum_{\{w\colon  |w|=|Cw_n|, w\preceq Cw_n \}} f_e(Bw) \\
&= \sum_{\{v\colon |v|=|B|,~v\prec B\}}\sum_{|w|=|Cw_n|} f_e(vw) + \sum_{\{w\colon  |w|=|Cw_n|, w\preceq Cw_n \}} f_e(w). 
\end{align*}
Using Proposition \ref{thm:perron} again, there is some nonnegative rational $\kappa'$ and some positive constant $C_0'$ and some $\beta'\in (0,k)$ such that
$$\left| \sum_{\{v\colon |v|=|B|,~v\prec B\}}\sum_{|w|=|Cw_n|} f_e(vw) - \kappa' k^{|Cw_n|} \right| <C_0' (\beta')^{|w_n|},$$  
where
\begin{equation}\label{eq:kappaprime}
\kappa' = \sum_{\{v\colon |v|=|B|,~v\prec B\}} c_{\delta(e,v),j}\in \mathbb{Q}.
\end{equation}
By increasing either $\beta$ or $\beta'$, we may assume that $\beta=\beta'$.
Using this fact along with Equations (\ref{eq:s}) and (\ref{eq:1}), we see there is some $C_1>0$ such that 
\begin{equation}
\label{eq:ABC} \left| s([ABCw_n]_k) - (\kappa k^{|BC|} + \kappa' k^{|C|}) k^{|w_n|}  - s_e(Cw_n)\right| < C_1 \beta^{|w_n|}.
\end{equation}
We similarly get positive constants $C_2$ and $C_3$ such that 
\begin{equation}
\label{eq:AC} \left| s([ACw_n]_k) - \kappa k^{|C|} k^{|w_n|}  - s_e(Cw_n)\right| < C_2 \beta^{|w_n|}
\end{equation}
and
\begin{equation}
\label{eq:ABBC} \left| s([ABBCw_n]_k) - (\kappa k^{|BBC|} + \kappa' k^{|BC|} + \kappa' k^{|C|}) k^{|w_n|}  - s_e(Cw_n)\right| < C_3 \beta^{|w_n|}.
\end{equation}
We next let $r_n = [Cw_n]_k$ and we let $u=[B]_k$ and $v=[A]_k$.  
By assumption
$$s([ABCw_n]_k)/[ABCw_n]_k\to \gamma \qquad {\rm as}\qquad n\to \infty,$$ and hence Equation (\ref{eq:ABC}) gives
$$\left(\kappa k^{|BC|} k^{|w_n|} + \kappa' k^{|C|} k^{|w_n|} + s_e([Cw_n]_k)\right)/\left( vk^{|BCw_n|}+ uk^{|Cw_n|}+r_n\right)\to \gamma$$ as $n\to\infty$.
Now since $\limsup s([ACw_n]_k)/[ACw_n]_k \le \gamma$, we have
$$\limsup \left(\kappa k^{|C|} k^{|w_n|} + s_e([Cw_n]_k) \right)/(vk^{|Cw_n|}+r_n)\le \gamma.$$ 
Now we apply Lemma \ref{limsup} with
$u' =\kappa'$, $v'=\kappa$, $s_n'=s_e([Cw_n]_k)k^{-|w_n|}$, $b=|B|$, $c=|C|$, $s_n=r_n k^{-|w_n|}$ and $u,v$ as above to deduce that
$$\limsup \frac{\left(\kappa k^{|BBC|} k^{|w_n|} + \kappa' k^{|BC|} k^{|w_n|} + \kappa' k^{|C|} k^{|w_n|} +  s_e([Cw_n]_k) \right)}{\left(uk^{|BBCw_n|}+ vk^{|BCw_n|}+vk^{|Cw_n|}+r_n\right)}$$
is greater than or equal to $\gamma$, and so $\gamma_j(ABBC)= \gamma$, which gives the result.  
\end{proof}
As an immediate consequence we get the following.
\begin{proposition}
\label{prop:X}
Adopt the assumptions and notation of Notation \ref{notn:1} and Notation \ref{notn:2}. Suppose that $A, B, C\in \Sigma_k^*$ are non-trivial words whose lengths are multiples of $a$ such that $A$ has no leading zeros, $ABC\in \mathcal{S}$, $\delta(1,A)=\delta(1,AB)=\delta(1,ABC)$, and $\gamma_j(ABC)=\gamma$. Then $\gamma_j(AB^sC)=\gamma$ for every $s\ge 1$.
\end{proposition}
\begin{proof}
By assumption, the result holds for $s=1$ and Lemma \ref{lem:key} gives $\gamma_j(AB^2C)=\gamma$. Now suppose that the result holds whenever $s\le m$.  Then 
$\gamma_j(AB^m C)=\gamma$ and we may apply Lemma \ref{lem:key}, leaving $A$ and $B$ unchanged and using $B^{m-1}C$ for $C$ and we get
$\gamma_j(AB^2 B^{m-1}C)=\gamma$.  The result now follows by induction.
\end{proof}
\begin{proof}[Proof of the Theorem \ref{thm:main}] 
Without loss of generality we may scale the set $\Delta$ by a positive rational number and then assume that $\Delta\subseteq \mathbb{Q}\cap [0,1]$. In particular, the maps $f_i$ take values in $[0,1]\cap \mathbb{Q}$ for $i=1,\ldots ,d$. Applying Proposition \ref{prop:X}
to $A=U$, $B=V$, and $C=W$, we see $\gamma_j(AB^m)=\gamma$ for every $m\ge 0$. If we take $w$ to be a word whose length is congruent to $j$ mod $a$, then the argument in Lemma \ref{lem:key} shows that 
$$s([AB^mw]_k) = \kappa k^{|w|+m|B|} + \kappa' (1+k^{|B|}+ \cdots + k^{(m-1)|B|}) k^{|w|} + s_{\ell}([w]_k),$$

with $\kappa=\kappa_{A,j}$ and $\kappa'=\kappa'_{B,j}$ given as in Equations (\ref{eq:kappa}) and (\ref{eq:kappaprime}); that is,
\begin{equation}\label{eq:kappa1}
\kappa=\kappa_{A,j}=  \sum_{\{v\colon |v|=|A|,~v\prec A\}} c_{\delta(1,v),j}
\end{equation}
and 
\begin{equation} \label{eq:kappaprime1}
\kappa'=\kappa'_{B,j} = \sum_{\{v\colon |v|=|B|,~v\prec V\}} c_{\delta(\ell,v),j}.
\end{equation}
Then since $f_{\ell}$ takes values in $[0.1]$, we see that $0\le s_{\ell}([w]_k)\le k^{|w|}$.  

Hence \begin{equation} \kappa k^{m|B|} + \frac{\kappa' (k^{m|B|}-1)}{(k^{|B|}-1)} \le s([AB^m w]_k)/k^{|w|} \le \kappa k^{m|B|} + \kappa' \frac{(k^{m|B|}-1)}{(k^{|B|}-1)} + 1.\label{eq:ss}\end{equation}
On the other hand, 
 \begin{align*} [A]_k k^{m|B|} + [B]_k \frac{(k^{m|B|}-1)}{(k^{|B|}-1)} & \le  [AB^m w]_k/k^{|w|}\\
  & \le  [A]_k k^{m|B|} + [B]_k \frac{(k^{m|B|}-1)}{(k^{|B|}-1)} + 1.\label{eq:rr}\end{align*}

Combining these two equations, we see
$$\left| s([AB^m w]_k)/[AB^m w]_k - \left(\kappa + \kappa'/(k^{|B|}-1)\right)/\left([A]_k + [B_k]/(k^{|B|}-1)\right)\right| =  {\rm O}(1/k^{m|B|}),$$ where the implied constant in ${\rm O}$ is independent of $m$.  
Since $\kappa$ and $\kappa'$ are rational,
 \begin{equation}
 \label{eq:YY}
 \alpha:=\left(\kappa + \kappa'/(k^{|B|}-1)\right)/\left([A]_k + [B_k]/(k^{|B|}-1)\right)\end{equation} is a rational number.  We claim that $\gamma=\alpha$.  To see this, observe that for every $m$, there are words $w_{m,n}$, $n\ge 1$, such that $s([AB^m w_{m,n}]_k)/[AB^m w_{m,n}]_k \to \gamma$ as $n\to \infty$.  On the other hand, we have shown that there is a fixed constant $C>0$ such that 
 $$|s([AB^m w_{m,n}]_k)/[AB^m w_{m,n}]_k - \alpha | < C k^{-m|B|}.$$ Taking the limit as $n$ tends to infinity then gives that
 $|\gamma-\alpha| < C k^{-m|B|}$.   Since $|B|>0$ and since this holds for every $m>0$, $\gamma=\alpha$. 
 Finally, we note that $\gamma$ is recursively computable, since Proposition \ref{thm:perron} gives that $a$ and the rational constants $c_{i,j}$ are recursively computable and hence $\kappa_{A,j}$ and $\kappa'_{B,j}$ given in Equations (\ref{eq:kappa1}) and (\ref{eq:kappaprime1}) are recursively computable for each pair of words $(A,B)$ and each $j\in \{0,\ldots ,a-1\}$ and since the limsup of $s(n)/n$ is of the form $\gamma_i(AB^m)$ for every $m\ge 1$ with some pair of words $(A,B)$ with $|A|+|B|\le (2d+1)a$ and some $i\in \{0,1,\ldots, a-1\}$, we can proceed as follows.  The limit of $\gamma_i(AB^m)$ as $m\to\infty$ is recursively computable by Equation (\ref{eq:YY}), and letting $A$ and $B$ range over non-trivial words with $|AB|\le (2d+1)a$ and $A$ having no leading zeros and letting $i$ range over $0,\ldots ,a-1$, and computing these limits and taking the maximum of these values, we see we can compute the limsup of $s(n)/n$ as $n\to\infty$. 
 \end{proof}

 \begin{proof}[Proof of Corollary \ref{cor:main2}] Define $f:\mathbb{N}\to \{0,1\}$ and $g:\mathbb{N}\to \{0,1\}$ via the rules $f(n)=1$ if $n\in S$ and $f(n)=0$ if $n\not\in S$ and $g(n)=1$ if $n\not\in S$ and $f(n)=0$ if $n\in S$.  
We let $s(n)=\sum_{i<n} f(n)$ and we let $t(n)=\sum_{i<n} g(n)$.
 Then $\gamma:=\limsup s(n)/n$ and $\gamma':=\limsup t(n)/n$ are recursively computable rational numbers by Theorem \ref{thm:main}.  Notice that $\gamma$ is just the upper density of $S$.  On the other hand, $s(n)+t(n)=n$ and if we let $\beta$ denote the lower density of $S$ 
 then \begin{align*}
 &~ \beta = \liminf s(n)/n = - \limsup (-s(n)/n) \\ &= - \limsup (t(n)-n)/n = - \left( -1+ \limsup t(n)/n\right) = 1- \gamma',\end{align*} 
 which is a recursively computable rational number.  The result follows.
 \end{proof}
 \section{An algorithm for computing the limsup}\label{computable}
In this section, we give an algorithm to compute $\limsup s(n)/n$, where $s(n)$ is the $n$-th partial sum of a $\mathbb{Q}_{\ge 0}$-valued automatic sequence $h(n)$.  We note that this algorithm essentially falls out of the proof of Theorem \ref{thm:main}, but we record it here explicitly as it may be of interest to people working with a given automatic sequence or set.  Then there is an automatic map $f:\Sigma_k^*\to \mathbb{Q}_{\ge 0}$ satisfying $f(0w)=f(w)$ and $f((n)_k)=h(n)$.

We now use the assumptions and notation of Notation \ref{notn:1} and Notation \ref{notn:2}.  We assume that we have a DFAO that accepts $w\in \Sigma_k^*$ as input, reading left to right, and gives $f(w)$ as output.  From the DFAO, we can construct automatic sequences $f=f_1,\ldots ,f_d$ that make up the left $k$-kernel of $f$.  By Proposition \ref{thm:perron} there is a recursively computable natural number $a$ and recursively computable rational numbers $c_{i,j}$ with $1\le i\le d$, $0\le j<a$ such that
$$s_i((k^{an+j})_k)/k^{an+j} = c_{i,j} +{\rm o}(1).$$

The first step is compute $a$ and the values $c_{i,j}$.  Using this terminology we get that there are words $A$ and $B$ whose lengths are multiples of $a$ with $|A|+|B|\le (2d+1)a$ and some $\ell\in \{0,1,\ldots ,a-1\}$ such that
$$\limsup_n s(n)/n = \left(\kappa_{A,j} + \kappa'_{B,j}/(k^{|B|}-1)\right)/\left([A]_k + [B_k]/(k^{|B|}-1)\right)$$ where
$\kappa_{A,j}$ and $\kappa'_{B,j}$ are as given in Equations (\ref{eq:kappa1}) and (\ref{eq:kappaprime1}).  
Moreover, each of the values of this form, by construction, occurs as a limit point of the sequence $s(n)/n$.  Consequently, one can take the maximum of the numbers of the form $$\left(\kappa_{A,j} + \kappa'_{B,j}/(k^{|B|}-1)\right)/\left([A]_k + [B_k]/(k^{|B|}-1)\right),$$ as one lets $A$ and $B$ range over the set of non-trivial words whose lengths are multiples of $a$ with $|A|+|B|\le (2d+1)a$ such that $A$ has no leading zeros and lets $j$ range over $\{0,\ldots ,a-1\}$, and the maximum of these values 
 will be $\limsup_n s(n)/n$.

As an example of how one can apply this in practice, we let $h(n)$ be the $3$-automatic sequence whose value is $1$ if the most significant ternary digit of $n$ is equal to $1$ and is zero otherwise, and we let $s(n)$ denote the $n$-th partial sum of $\{h(j)\}_{j\ge 0}$. Then there is $f:\Sigma_3^*\to \{0,1\}$ with the property that $f(0w)=f(w)$ for $w\in \Sigma_3^*$ and $f((n)_k)=h(n)$.  If we look at the left kernel of $f$, it consists of the $f=f_1,f_2,f_3$ where $f_2$ is the constant function $0$ and $f_3$ is the constant function $1$, and we have the rules
$f_1([0w]_k)=f_1([w]_k)$, $f_1([1w]_k)=f_3([w]_k)$, $f_1([2w]_k)=f_2([w]_k)$ and $f_2([iw]_k)=f_2([w]_k)$ and $f_3([iw]_k)=f_3([w]_k)$ for 
$w\in \{0,1,2\}^*$ and $i\in \{0,1,2\}$.  In this case, we find we can take $a=1$ in Proposition \ref{thm:perron} and that 
$f_1(3^n) \sim (1/2)\cdot 3^n$, $f_2(3^n)=0$, $f_3(3^n)=3^n$. By the algorithm described above, the limsup of $s(n)/n$ is the maximum over words $A,B\in \{0,1,2\}^*$ such that $A$ has no leading zeros and
$|A|+|B|\le 6$.  We let $i\in \{1,2\}$ denote the first letter of $A$ and we write $A=iA'$. Using Equations (\ref{eq:kappa1}) and (\ref{eq:kappaprime1}) we see
$$\kappa_{A,0}=k^{|A|-1}/2 +\delta_{i,2} k^{|A|} +\delta_{i,1}[A']_k,$$ and $$\kappa_{B,0}'=\delta_{i,1}[B]_k,$$ regardless of what $A'$ and $B$ are.  Thus Equation (\ref{eq:YY}) gives $\gamma_0(AB^m)$ tends to 
$$(\kappa_{A,0}+\kappa_{B,0}'/(k^{|B|}-1))/([A]_k + [B]_k/(k^{|B|}-1))$$ as $m\to\infty$.
Checking these values for the allowable $A$ and $B$, we see this is maximized when $A$ is the one-letter word $1$ and $B$ is the one-letter word $2$, in which case one gets a limit of $3/4$.  
 \section{Proof of Theorem \ref{thm:main3}}
 \label{example}
In this section we give the proof Theorem \ref{thm:main3}.  Given a set of natural numbers $S$, we let $\pi_S(x)$ denote the number of elements in $S$ that are less than $x$.  To get the final part of this characterization we need a simple lemma.
 
 \begin{lemma} Let $k\ge 2$ be a natural number and let $S$ be a $k$-automatic set.  Then 
$$\liminf_{N\to \infty} \pi_S(N)/N =0\implies \limsup_{N\to\infty} \pi_S(N)/N=0$$ and 
$$\limsup_{N\to \infty} \pi_S(N)/N =1\implies \liminf_{N\to\infty} \pi_S(N)/N=1.$$
\label{lem:01}
\end{lemma}
\begin{proof}
By Proposition \ref{thm:perron}, there exist $\beta\in (0,k)$, a positive integer $a$, and nonnegative rational numbers $c_0,\ldots ,c_{a-1}$ such that
$\pi_S(k^{an+j})= c_j k^{an+j}+{\rm O}(\beta^{an})$. We claim that if
$\liminf_{N\to \infty} \pi_S(N)/N=0$ then $c_0,\ldots ,c_{a-1}$.  To see this, suppose that this is not the case.  Then there is some $i$ such that $c_i>0$.  Then for every $N$ sufficiently large there is some $n$ such that
$k^{an+i} \le N < k^{a(n+1)+i}$.  Hence 
$\pi_S(N)/N \ge \pi_S(k^{an+i})/k^{a(n+1)+i} \sim c_i/k^a$ as $n\to\infty$.  It follows that if some $c_i>0$ then $\liminf \pi_S(N)/N>0$, and so the claim follows.  Thus if $\liminf_{N\to \infty} \pi_S(N)/N=0$ then we have $c_0=\cdots = c_{a-1}=0$. We claim that this then gives that $\limsup \pi_S(N)/N = 0$.  To see this, suppose that $\limsup \pi_S(N)/N =\gamma >0$.  Then there are infinitely many $N$ such that $\pi_S(N)> \gamma N/2$.  For such $N$ we have an $n$, depending upon $N$, such that $k^{an} \le N < k^{an+a}$ and so
$$\pi_S(k^{an+a})/k^{an+a} > \pi_S(N)/(k^a N) > \gamma/(2k^a),$$ contradicting the fact that $c_0=0$.  Thus we have shown that 
$$\liminf \pi_S(N)/N = 0 \implies \limsup \pi_S(N)/N=0.$$
Similarly, if we let $T$ denote the complement of $S$ then $T$ is automatic and $\limsup \pi_S(N)/N = 1$ if and only if $\liminf \pi_T(N)/N=0$, and so  if $\limsup \pi_S(N)/N=1$ then $\liminf \pi_T(N)/N=0$ and hence $\limsup \pi_T(N)/N=0$ and thus $\liminf \pi_S(N)/N=1$.  This completes the proof.\end{proof}
 
 \begin{proof}[Proof of Theorem \ref{thm:main3}] By Lemma \ref{lem:01} and Corollary \ref{cor:main2}, if $S$ is a $k$-automatic having lower density and upper density $\alpha$ and $\beta$ respectively then either $(\alpha,\beta)\in \{(0,0),(1,1)\}$ or $\alpha,\beta$ are rational numbers with $0<\alpha\le \beta<1$.  We notice that the empty set has upper and lower density $0$ and $\mathbb{N}$ has upper and lower density $1$, and so to complete the proof, it suffices to show that whenever $\alpha,\beta$ are rational numbers satisfying $0<\alpha \le \beta <1$ then there is a $k$-automatic set whose lower density is $\alpha$ and whose upper density is $\beta$.  Since a set is $k$-automatic if and only if it is $k^m$-automatic for each positive integer $m$, we may replace $k$ with a power and assume that $k\alpha\ge \beta$ and that $k\beta < (k-1)$.  We let
$$\alpha'= (k\alpha -\beta)/(k-1)\qquad {\rm and} \qquad \beta' =(k\beta -\alpha)/(k-1).$$  Since $k\alpha>\beta\ge \alpha$, we see $\alpha',\beta'>0$.  Also, since $k\beta, k\alpha<k-1$ we have $\alpha',\beta'< 1$.
  We have
 \begin{equation}
 k\alpha' + \beta'=(k+1)\alpha\qquad {\rm and}\qquad k\beta'+\alpha'=(k+1)\beta.
 \end{equation}
  Now let $A$, $B$, and $C$ be positive integers with $C>A,B$ and $A/C=\alpha'$ and $B/C=\beta'$.  
  We let $T_0$ be the set of natural numbers that  
 are either $0,1,\ldots ,A-2$, or $A-1$ mod $C$; we let $T_1$ be  the set of natural numbers that  
 are either $0,1,\ldots ,B-2$, or $B-1$ mod $C$.  We let $U_0$ be the set of natural numbers whose base-$k$ expansion has even length and we let $U_1$ be the set of natural numbers whose base-$k$ expansion has odd length.  We let $S=(U_0\cap T_0)\cup (U_1\cap T_1)$. Since each of $U_0,U_1,T_0,T_1$ are $k$-automatic sets, so is $S$ as such sets are closed under finite intersections and unions. For an interval $I\subseteq [k^{2n},k^{2n+1})$ the number of elements in $S\cap I$ satisfies
\begin{equation}
 \frac{B\cdot \#(I\cap \mathbb{N})}{C} -C \le \#(S\cap I) \le \frac{B\cdot  \#(I\cap \mathbb{N})}{C} +C \end{equation}
 Similarly, for an interval $I\subseteq [k^{2n+1},k^{2n+2})$ the number of elements in $S\cap I$ satisfies
\begin{equation}
 \frac{A\cdot \#(I\cap \mathbb{N})}{C} -C \le \#(S\cap I) \le \frac{A\cdot \#(I\cap \mathbb{N})}{C}  +C \end{equation}
 From these inequalities, it is straightforward to deduce that $\pi_S(k^{2n})$ is asymptotic to
 $$(A/C)(1+(k^2-k)  +\cdots + (k^{2n}-k^{2n-1}))+(B/C)((k-1)+(k^3-k^2)+\cdots + (k^{2n-1}-k^{2n-2})),$$
 which gives
  $$(A/C)(k^{2n+1}+1)/(k+1) + (B/C) (k^{2n}-1)/(k+1).$$
 
From this we obtain the asymptotic result
\begin{equation}
\pi_S(k^{2n})\sim \alpha' k^{2n+1}/(k+1) + \beta' k^{2n}/(k+1) = \alpha k^{2n}.\end{equation} Similarly,
\begin{equation}
\pi_S(k^{2n+1})\sim \beta' k^{2n+2}(k+1) + \alpha' k^{2n+1}/(k+1) =\beta k^{2n+1}\end{equation} as $n\to \infty$.
It follows that for $N\in [k^{2n+1}+1,k^{2n}]$ we have
$$\pi_S(N)/N \sim \beta k^{2n+1}/N + \beta' - \beta' k^{2n+1}/N$$ as $n\to\infty$.
If we fix $n$ and let $N$ range over the interval $ [k^{2n+1}+1,k^{2n}]$, then since the function 
$\beta k^{2n+1}/x + \beta' - \beta' k^{2n+1}/x$ has derivative of the form $\kappa x^{-2}$, it is monotonic on the interval $ [k^{2n+1},k^{2n+2}]$
and hence the maximum and minimum are attained at the end points.  When $N=k^{2n+1}$, $\beta k^{2n+1}/N + \beta' - \beta' k^{2n+1}/N$ is equal to $\beta$, and at $N=k^{2n+2}$ it is equal to 
$\beta/k + \beta'(1-1/k)=\alpha$ and so 
$$\alpha = \liminf_{n\to\infty} \inf_{N\in [k^{2n+1},k^{2n}]} \pi_S(N)/N$$ and
$$\beta = \limsup_{n\to\infty} \sup_{N\in [k^{2n+1},k^{2n}]} \pi_S(N)/N.$$
Similarly,  
$$\alpha = \liminf_{n\to\infty} \inf_{N\in [k^{2n},k^{2n+1}]} \pi_S(N)/N$$ and
$$\beta = \limsup_{n\to\infty} \sup_{N\in [k^{2n},k^{2n+1}]} \pi_S(N)/N.$$
It follows that $\alpha$ and $\beta$ are respectively the liminf and limsup of $\pi_S(N)/N$, as claimed.  
\end{proof}
\section{Concluding remarks}\label{conc}
We have established that the upper and lower densities of an automatic set are recursively computable rational numbers.  It is natural to ask whether similar results hold when one looks at larger classes of sets.  A generalization of automatic sets is sets associated to a morphic word on the alphabet $\{0,1\}$.  A result of Cobham \cite{Cobham} shows that automatic sets are precisely those corresponding to \emph{uniform} morphisms, and so these morphic sets form a strictly larger class. In this case, the density of morphic sets need not be rational, if it exists.  
For example, the Fibonacci word, which is the right-infinite word $01001\cdots$ that is the unique fixed point of the morphism $0\mapsto 01$ and $1\mapsto 0$ whose first letter is $0$, corresponds to the set $\{1,3,4,\ldots \}$ (i.e., the elements of the set are the positions where the $0$'s occur in the sequence).  This set has density $1/\rho^2$, where $\rho=(\sqrt{5}+1)/2$ \cite[Proposition 2.1.10]{L}.  We suspect that the densities, when they exist, and the upper and lower densities of morphic sets should be algebraic numbers.  We believe the techniques in this paper along with those given in \cite{Bell} might be useful in establishing these facts, although we leave this as a question for others to think about.

Another interesting generalization of automatic sets are those produced via push-down automata; i.e., context-free subsets of $\mathbb{N}$, where we once again assume that the subset is formed by taking the natural numbers whose base-$k$ expansions form a context-free sublanguage of $\{0,1,\ldots ,k-1\}^*$ for some $k\ge 2$.  Here there is an interesting dichotomy that arises: unambiguous context-free and ambiguous context-free.  The former case is much better behaved and work of Chomsky and Sch\"utzenberger \cite{CS}, along with basic asymptotic results for algebraic functions, shows that if the density of such a set  of natural numbers exists then it is necessarily an algebraic number.  On the other hand, the question of whether upper and lower densities of unambiguous context-free subsets of $\mathbb{N}$ are algebraic or not is apparently open.  For ambiguous context-free sets, the behaviour can be much more pathological. Work of Kemp \cite{Kemp}, with some small additional amount of arguing, shows that such densities can be transcendental.

 \section*{Acknowledgments}
 I thank Jeffrey Shallit for bringing this problem to my attention and for making many useful comments and for bringing my attention to the work of Kemp.  I also thank the referee for reading the paper carefully and giving numerous helpful comments and suggestions. 
 
\end{document}